\documentclass{article}[12pt]
\usepackage{indentfirst}
\usepackage{tabularx}
\usepackage[a4paper]{geometry}
\usepackage[utf8]{inputenc}
\usepackage[T1]{fontenc}
\usepackage{lmodern}
\usepackage{amssymb}
\usepackage{amsfonts}
\usepackage{amscd}
\usepackage{amsmath,amssymb,amsthm} \allowdisplaybreaks[1]
\usepackage{mathrsfs}
\usepackage{microtype}
\usepackage{latexsym}
\usepackage[colorlinks]{hyperref}
\usepackage{cleveref}
\usepackage{setspace}
\usepackage{graphicx}
\usepackage{bookmark}
\usepackage{booktabs}
\usepackage{diagbox}
\usepackage{float}
\usepackage{color}   %\cite{paper}  \ref{}  \eqref{equation}
\usepackage{bm}
\usepackage{longtable}
\usepackage{xcolor}
\usepackage{makecell}
\usepackage{enumerate}
\usepackage{fullpage}
\usepackage{ulem}

\newtheorem{thm}{Theorem}[section]

\newtheorem{prop}{Proposition}[section]
\newtheorem{cor}{Corollary}[section]
\theoremstyle{definition}
\newtheorem{de}{Definition}[section]

\newtheorem{re}{Remark}[section]
\newtheorem{ex}{Example}[section]
\newcommand{\C}{\mathbb{C}}

\newcommand{\Q}{\mathbb{Q}}
\newcommand{\Z}{\mathbb{Z}}

\newcommand{\be}{\begin{eqnarray}}
\newcommand{\ee}{\end{eqnarray}}
\newcommand{\nn}{{\nonumber}}
\newcommand{\dd}{\displaystyle}
\newcommand{\ra}{\rightarrow}

\title{ Butson Hadamard matrices, bent sequences, and spherical codes }
\author{ Minjia Shi, Danni Lu\thanks{Minjia Shi and Danni Lu are with the Key Laboratory of Intelligent Computing Signal
		Processing, Ministry of Education, School of Mathematical Sciences, Anhui
		University, Hefei 230601, China; State Key Laboratory of integrated Service Networks, Xidian University, Xi'an,
		710071, China. smjwcl.good@163.com, ludanni\_in@163.com},  Andr\'es Armario \thanks{Departamento de Matemática Aplicada I, Universidad de Sevilla, Avda. Reina Mercedes s/n, 41012, Seville, Spain. armario@us.es This research was supported by the Strategic R+D Project TED2021-130566B-I00 from the Ministry of Science and Innovation of the Government of Spain.}, 
	Ronan Egan \thanks{School of Mathematical Sciences, Dublin City University, Ireland. ronan.egan@dcu.ie}, 
	Ferruh Ozbudak \thanks{Ferruh \"{O}zbudak is with Faculty of Engineering and Natural Sciences, Sabanc{\i} University, 34956, Istanbul,  E-mail: ferruh.ozbudak@sabanciuniv.edu}, 
	Patrick Sol\'e \thanks{{ Patrick Sol\'e is with I2M (Aix Marseille Univ, CNRS, Centrale Marseille), Marseilles, France. sole@enst.fr}}
}
\date{}
\begin{document}
	\maketitle
	\begin{abstract}
		We explore a notion of bent sequence attached to the data consisting of an Hadamard matrix of order $n$  defined over the complex $q^{th}$ roots of unity, an eigenvalue of that matrix, and a Galois automorphism from the cyclotomic field of order $q.$ In particular we construct self-dual bent sequences for various $q\le 60$ and lengths $n\le 21.$
		Computational construction methods comprise the resolution of polynomial systems by Groebner bases and eigenspace computations. Infinite families can be constructed from regular Hadamard matrices, Bush-type Hadamard matrices, and generalized Boolean bent functions.
		As an application, we estimate the covering radius of the code attached to that matrix over $\Z_q.$ We derive a lower bound on that quantity
		for the Chinese Euclidean metric when bent sequences exist. We give the Euclidean distance spectrum, and bound above the covering radius of an attached spherical code, depending on its strength as a spherical design.
		
	\end{abstract}
	\textbf{Keywords:} Hadamard matrices, Butson matrices, covering radius, bent sequences, spherical codes\\ \ \
{\textbf MSC (2020):} Primary 05 B20, Secondary 05E99
	\section{Introduction}\label{Introduction}
Bent sequences are important combinatorial objects that are relevant to difference sets, Boolean functions,  strongly regular graphs, and symmetric cryptography \cite{M}. They are classically defined in relation to the Sylvester matrix, the matrix of the Walsh-Hadamard transform.
Recently a new notion of bent sequences, motivated by a cryptographic problem (PUFs=Physically Unclonable Functions), appeared, that is defined for any Hadamard matrix \cite{rioul}. Constructions and properties of these sequences in the self-dual case were investigated in \cite{S+}.
 The advantage of self-dual bent sequences over unrestricted bent sequences
 is that they allow linear algebra to bear on the problem. Note that every bent sequence is self-dual for some matrix \cite{ML}.
 The extension of this notion to
complex Hadamard matrix in the sense of Turyn, that is to say with entries in $\Omega_4=\{\pm 1, \pm i\}$, was done in \cite{ML}. A self-dual bent sequence $X$ attached to a matrix $H$ is then an eigenvector of $H$ for an eigenvalue $\lambda$ of $H$ (some Gaussian integer) that have its entries in $\Omega_4.$ Thus $HX=\lambda X.$ A new twist in the definition
of self-dual bent sequences was given in \cite{RI} were it is suggested to replace the above equation by $HX=\lambda \overline{X},$ where the bar denotes complex conjugation. The motivation for that change was to obtain more bent sequences, in particular in relation with generalized bent functions.

In the present paper, we extend the notion of self-dual bent sequences in two directions. First, we consider Butson Hadamard matrices, that is to say Hadamard matrices with values in the complex roots of unity of a given order. Next, we replace the notion of conjugate by that of multiplier, that is a bijective map of the form $z\mapsto z^k$ with $(k,q)=1$ which preserves globally the complex $q^{th}$ roots of unity. When $k=1$ this new definition reduces to that of \cite{ML,S+}. When $k=q-1$, we recover the definition of \cite{RI}. This change of definition affects the notion of a strong group, a matrix group that preserves the set of self-dual bent sequences.
We extend the computational techniques (linear algebra and Groebner bases) from \cite{ML,S+} to this new setting. This allows us to obtain many examples of self-dual bent sequences from the matrices of the database \cite{W}. General constructions of Butson Hadamard matrices with designed self-dual bent sequences are then derived from, respectively, Kronecker products of Butson matrices, regular Butson matrices, Bush type Hadamard  matrices, and generalized Walsh Hadamard
matrices. In the last case we use group invariant Butson matrices in the sense of \cite{TD,Sch19}. Strong necessary existence conditions for Butson matrices admitting self-dual bent sequences with trivial multiplier are given. In a second part of the article, we turn our attention towards Butson-Hadamard codes, and their metric properties with respect to covering and packing either a finite set of strings or the Euclidean sphere. While the packing properties of the code with respect to the homogeneous and other metrics have been investigated by various
authors \cite{T,G}, we focus on the Chinese Euclidean metric and the covering properties in connection with bent sequences. The Chinese Euclidean metric, which is first mentioned in \cite{CE}, should not be confused with that in \cite{JL}.
We also associate with every such code of length $n$ a spherical code in dimension $2n.$ This code is shown to be optimal for the minimum distance in the case of $q=4,$ as meeting the Levenshtein bound. Pending on some properties of the Hadamard matrix, we can obtain spherical designs of strength one or two. Spherical designs are finite sets of points on the unit Euclidean sphere which occur naturally in relation to the ``cubature formulae'' of numerical analysis,
 and can be thought of as analogues of combinatorial designs \cite{DGS}. In \cite{PD}, the existence of spherical designs is established, providing a solid foundation for their theoretical exploration. Moreover, the utilization of spherical designs in motivating the definition of conformal $t$-designs is demonstrated in \cite{GH}. Building spherical designs from Butson Hadamard matrices seems to be a new approach. It enters here as a way to use upper bounds on the covering radius of spherical designs \cite{FL,S} to control the covering radius of the said spherical code.

The material is arranged as follows. Section 2 collects the definitions needed for the rest of the paper. Section 3 defines and constructs the strong automorphism group of a Butson matrix with a given multiplier. Section 4 studies computational methods and results. Section 5 gives four general constructions of Butson Hadamard matrices with designed self-dual bent sequences. Section 6 derives a necessary condition for the existence of Butson matrices admitting a self-dual bent sequence. Section 7 derives the Euclidean distance spectrum of Butson Hadamard codes. Section 8 contains the lower and upper bounds on the covering radius of Butson Hadamard codes and attending spherical codes. Section 9 concludes the article.
	\section{Preliminaries}\label{Pre}
	{\de A complex {\bf Hadamard} matrix of order $n$ is {an}  $n \times n$ matrix $H$ {with entries in the complex unit  satisfying the equation $HH^*=nI,$ where the * denotes the transpose conjugate. If all of its elements are in $\Omega_q=\left \{ z\in C\mid z^q=1 \right \},$ for some integer $q, $ then $H$ is said to be {\bf Butson} type, and we write  $H\in BH(n,q).$ A matrix $H \in BH(n,q)$ is in {\bf dephased} form when both first row and first columns only contain ones.}

{\de Two Hadamard matrices $H$ and $K$ of $BH(n,q)$ are called {\bf equivalent} if there are two permutation matrices $P_1$ and $P_2$ and two diagonal matrices $D_1$ and $D_2$ with entries in $\Omega_q$ such that
$$H=P_1D_1 K P_2D_2.$$
Every matrix of $BH(n,q)$ is equivalent to its dephased form.}
	{\de Given $H\in BH(n,q)$ the sequence $X\in\Omega_q^n$ is a {\bf bent sequence} if there is $Y\in\Omega_q^n$ such that  $HX=\lambda Y$ for some $\lambda \in \Z[\zeta_q].$ 
}
	{\de \label{definition.self-dual.bent} For any $k$ coprime with $q$ we define the {\bf multiplier} $\mu_k:\Omega_q \longrightarrow\Omega_q $ by the rule $\mu_k(z)=z^k.$
		Now a self-dual bent sequence attached to $H\in BH(n,q)$ is defined as $X\in\Omega_q^n$ such that $$HX=\lambda\mu_k(X),$$ where $\lambda$ is an element of $\Z[\zeta_q].$ (When $k=1$ it is also an eigenvalue of $H$).
		When $k=1,$ this is the definition of self-dual bent sequence in \cite{ML}.
		When $k=q-1,$ this is the definition of self-dual bent sequence in \cite{RI}. }\\ \ \

{\re \noindent
\begin{enumerate}
\item[(1)] Any multiplier $\mu_k$ can be extended into a Galois automorphism of the cyclotomic field of index $q.$ Writing an element of that field as
$a=\sum\limits_{i=0}^{q-1} a_i \zeta^i$ for rationals $a_i$'s, we define $\mu_k(a)=\sum\limits_{i=0}^{q-1} a_i \zeta^{ki}.$ In fact all Galois automorphisms are of that form \cite{L}.
They thus satisfy the following three properties.
\begin{enumerate}[(i)]
\item If $x\in \Q,$ then $\mu_k(x)=x.$
\item $\forall x,y \in \Q(\zeta_q),\, \mu_k(x+y)=\mu_k(x)+\mu_k(y).$
\item $\forall x,y \in \Q(\zeta_q),\, \mu_k(xy)=\mu_k(x)\mu_k(y).$
\end{enumerate}
\item[(2)] If $X$ is a bent sequence for $H\in BH(n,q)$ with $HX=\lambda Y,$ then $Y$ is a self-dual bent sequence with trivial multiplier for the matrix $HD,$ where $D$ is a diagonal matrix such that $X=DY.$ Note that the diagonal of $D$ lies in $\Omega_q,$ which implies that $HD \in BH(n,q).$
	\end{enumerate}}

	{\de The {\bf Chinese Euclidean weight} $w_{CE}(x)$ of a vector $x\in \mathbb{Z}_q^n$ is
		$$\sum_{i=1}^{n}\left[ 2-2\cos(\frac{2\pi x_i}{q} ) \right],  $$}
and the {\bf Chinese Euclidean distance} between the codewords $u$ and $v\in \mathbb{Z}_q^n$ is defined as $$d_{CE}(u,v)=w_{CE}(u-v).$$

	{\de A Butson matrix $H\in BH(n,q)$ is conveniently represented in logarithmic form, that is, the matrix $H=\left [ \zeta_q^{\varphi_{i,j}} \right ]_{i,j=1}^{n}$ is represented by the matrix $L(H)=\left [ \varphi_{i,j} \bmod q  \right ]_{i,j=1}^{n}$ with the convention that $L_{i,j}\in \mathbb{Z}_q$ for all $i,j\in \{1,2,\dots,n\}.$ Given $H\in BH(n,q),$ we denote by $F_H$ the $\mathbb{Z}_q$-code of length $n$ consisting of the rows of $L(H),$ and by $C_H=\cup_{\alpha\in \mathbb{Z}_q}(F_H+\alpha \textbf{1})$ where \textbf{1} denotes the all-one vector. The code $C_H\subseteq \mathbb{Z}_q^n$ is called a {\bf Butson Hadamard} code. Any subset of $\Omega_q^n$ for some integer $n$ is called a {\bf polyphase} code, following \cite{EZ}. }
	
	{\de Let $C_H$ be a Butson Hadamard code of length $n$ over $\mathbb{Z}_q$ coming from $H \in BH(n,q)$. The {\bf covering radius} of $C_H$ with Chinese Euclidean distance $d_{CE}$ is defined as: $$r_{CE}(C_H)=\max\{d_{CE}(x,C_H)]\mid x \in \mathbb{Z}_q^n\}=\max\limits_{x \in \mathbb{Z}_q^n}\min\limits_{y \in C_H}d_{CE}(x,y).$$
	}
\section{Automorphism group}

Denote the set of $n \times n$ monomial matrices with nonzero entries in $\Omega_{q}$ by $M(n,q)$. The set $M(n,q)$ is a group generated by the set $P_{n}$ of all permutation matrices of order $n$, and $\Delta(n,q)$, the group of all diagonal matrices of order $n$ with entries in $\Omega_{q}$. The group of all pairs of matrices from $M(n,q)$ acts on the set $\mathrm{BH}(n,q)$. The action is given by
\[
(P,Q)H = PHQ^{\ast}
\]
for all $H \in \mathrm{BH}(n,q)$. The orbit of $H$ under this action is the equivalence class of $H$. The stabilizer of $H$ under this action, i.e., the set of all $(P,Q) \in M(n,q)^{2}$ such that
\[
PHQ^{\ast} = H
\]
is called the {\bf automorphism group} of $H$, and is denoted $Aut(H)$. The subgroup of $Aut(H)$ comprised of pairs of matrices in $P_{n}$ is called the {\bf permutation automorphism group} of $H$, denoted $PAut(H)$.

We define the {\bf strong automorphism group} with trivial multiplier of $H$ to be the subgroup $SAut(H,1) \leq Aut(H)$ comprised of pairs of the form $(P,P)$. It follows from the definition of an automorphism that $(P,P) \in SAut(H,1)$ if and only if $PH = HP$. Hence $SAut(H,1)$ is in bijection with the centralizer of $H$ in $M(n,q)$.

Calculating the automorphism groups of $H \in \mathrm{BH}(n,q)$ is an onerous task in general. The most efficient approaches all involve a similar idea - constructing an appropriate larger $(0,1)$-matrix and using graph theoretic tools such as {\sc nauty} \cite{N} which are highly efficient. In the following section we condense the relevant details on computing the automorphism group of a pairwise combinatorial design given in \cite[Chapter 9]{deLF} down for just the purpose of handling a Butson Hadamard matrix.

\subsection{The expanded and associated designs}

{\de Let $H \in \mathrm{BH}(n,q)$. The {\bf expanded design} $\mathcal{E}_H$ of $H$ is the $nq \times nq$ matrix of the form
%\begin{equation}\label{exp des}
\[
\mathcal{E}_{H} = [\zeta_{q}^{i}H\zeta_{q}^{j}]_{0 \leq i,j \leq q-1}.
\]}
%\end{equation}
The {\bf associated design} $\mathrm{A}_{H}$
of $H$ is obtained from the expanded design by replacing each of its non-identity
entries with $0$.
\begin{ex}\label{expanded BH}
\upshape{Let $\zeta = \zeta_{3}$. If
\[H = {\small \renewcommand{\arraycolsep}{.1cm}\left[\begin{array}{rrr}
1 & 1 & 1 \\
1 & \zeta & \zeta^2 \\
1 & \zeta^2 & \zeta \end{array}\right]}
\]
then $\mathcal{E}_{H}$ and $\mathrm{A}_{H}$ are}
\[
{\footnotesize
\renewcommand{\arraycolsep}{.1cm}\left[\begin{array}{ccccccccc}
1 & 1 & 1 & \zeta & \zeta & \zeta & \zeta^2 & \zeta^2 & \zeta^2 \\
1 & \zeta & \zeta^2 & \zeta & \zeta^2 & 1 & \zeta^2 & 1 & \zeta \\
1 & \zeta^2 & \zeta & \zeta & 1 & \zeta^2 & \zeta^2 & \zeta & 1 \\
\zeta & \zeta & \zeta & \zeta^2 & \zeta^2 & \zeta^2 & 1 & 1 & 1 \\
\zeta & \zeta^2 & 1 & \zeta^2 & 1 & \zeta & 1 & \zeta & \zeta^2 \\
\zeta & 1 & \zeta^2 & \zeta^2 & \zeta & 1 & 1 & \zeta^2 & \zeta \\
\zeta^2 & \zeta^2 & \zeta^2 & 1 & 1 & 1 & \zeta & \zeta & \zeta \\
\zeta^2 & 1 & \zeta & 1 & \zeta & \zeta^2 & \zeta & \zeta^2 & 1 \\
\zeta^2 & \zeta & 1 & 1 & \zeta^2 & \zeta & \zeta & 1 & \zeta^2 \\
\end{array}\right]}~\textit{and}~
%\]
%\textit{and}
%\[
{\footnotesize
\renewcommand{\arraycolsep}{.1cm}\left[\begin{array}{ccccccccc}
1 & 1 & 1 & 0 & 0 & 0 & 0 & 0 & 0 \\
1 & 0 & 0 & 0 & 0 & 1 & 0 & 1 & 0 \\
1 & 0 & 0 & 0 & 1 & 0 & 0 & 0 & 1 \\
0 & 0 & 0 & 0 & 0 & 0 & 1 & 1 & 1 \\
0 & 0 & 1 & 0 & 1 & 0 & 1 & 0 & 0 \\
0 & 1 & 0 & 0 & 0 & 1 & 1 & 0 & 0 \\
0 & 0 & 0 & 1 & 1 & 1 & 0 & 0 & 0 \\
0 & 1 & 0 & 1 & 0 & 0 & 0 & 0 & 1 \\
0 & 0 & 1 & 1 & 0 & 0 & 0 & 1 & 0 \\
\end{array}\right]}
\]
respectively.
%Writing out the associated design adds nothing!
\end{ex}

The following is a special case of \cite[Theorem 9.6.12]{deLF}.
\begin{thm}\label{Iso Exp}
%If $H$ is a$\GH(n,\Crm_p)$ then
$Aut(H) \cong PAut(\mathcal{E}_{H})$.
 %\cap\Crm_{\mathrm{Perm}(np^m)^2}(Z_{\Crm_{p}^{m}})$.
\end{thm}

The isomorphism of Theorem~\ref{Iso Exp} is described in detail in
\cite[Section 9.6]{deLF}, we summarize the main points here.
For any $X \in M(n,q)$ there are unique disjoint $(0,1)$-matrices $X_{\omega}$
such that $X = \sum_{\omega \in \Omega_{q}}\omega{X_{\omega}}$.  Let
\begin{equation}\label{S and T}
S_{\omega} = [\delta_{\omega \beta}^{\alpha}]_{\alpha,\beta \in \Omega_{q}} \quad \text{ and }
\quad T_{\omega} = [\delta_{\beta}^{\alpha \omega}]_{\alpha,\beta \in \Omega_{q}}.
\end{equation}
Let $\theta^{(1)}(X) = \sum_{\omega \in \Omega_{q}}T_{\omega}\otimes X_{\omega}$ and
$\theta^{(2)}(X) = \sum_{\omega \in \Omega_{q}}S_{\omega}\otimes X_{\omega}$.
Define the map $\Theta : M(n,q)^{2} \rightarrow P_{nq}^{2}$
by $\Theta: (X,Y) \mapsto (\theta^{(1)}(X),\theta^{(2)}(Y))$. Let $Z_{q} = \{(S_{\omega} \otimes I_{n} , T_{\mu} \otimes I_{n}) \, \mid \, \omega,\mu \in \Omega_{q}\}$. Then by \cite[Theorem 9.6.7]{deLF}, $\Theta(M(n,q)^{2}) = C_{P_{nq} \times P_{nq}}(Z_{q})$, the centralizer of $Z_{q}$ in $P_{nq}\times P_{nq}$. The restriction $\Theta : Aut(H) \rightarrow PAut(\mathcal{E}_H)$ is an isomorphism as in Theorem \ref{Iso Exp}.
This leads to the following useful result.
\begin{thm}[Theorem 9.8.5, \cite{deLF}]
$\Theta(Aut(H)) = PAut(A_{H})\cap C_{P_{nq} \times P_{nq}}(Z_{q})$.
\end{thm}

Treating $A_{H}$ as the incidence structure of points and blocks, built in {\sc Magma} \cite{Magma} functions construct the automorphism group $PAut(A_{H})$.

\subsection{Automorphisms of digraphs}

Inspired by \cite[\S 3.1]{ML}, we attach a digraph $G(H)$ on $2qn$ vertices. 

%\begin{itemize}\marginpar{\tiny \textcolor{blue}{Andrés: I think we should remove statements 1 to 3. Repetitive}}
%\item The row $t$ of $H$ yields $q$ row vertices $r(t,x)$ and $q$ rows arcs $(r(t,x),\zeta_q r(t,x))$ for $x \in \Omega_q$
%\item The column $s$ of $H$ yields $q$ column vertices $c(t,x)$ and $q$ column arcs $(c(s,x),\zeta_q c(s,x))$ for $x \in \Omega_q$
%\item The entry $(t,s)$ yields $q$ arcs $(r(t,x), c(t,H_{rs}x))$ for $x \in \Omega_q$
%\end{itemize}

\begin{enumerate}[(1)]
    \item For each row $t$ of $H$, there are $q$ row vertices $r(t,x)$ where $x \in \Omega_{q}$, and $q$ arcs $(r(t,x),r(t,\zeta_{q}x))$ for each $x \in \Omega_{q}$.
    \item For each column $s$ of $H$, there are $q$ column vertices $c(s,x)$ where $x \in \Omega_{q}$, and $q$ arcs $(c(s,x),c(s,\zeta_{q}x))$ for each $x \in \Omega_{q}$.
    \item The entry $H_{ts}$ yields $q$ arcs $(r(t,x),c(s,H_{ts}x))$ for each $x \in \Omega_{q}$.
\end{enumerate}

The analogue of \cite[Corollary 1]{ML} is stated below.  The proof is given in the proof of Proposition \ref{prop_Saut}.

{\prop \label{prop_aut} The automorphism group $Aut(H)$ is isomorphic to the automorphism group of $G(H).$ }

We now consider the strong automorphism group of $C.$
%Given a multiplier $\mu_k$ one may define its fixed set as
%$$S_k=\{ z \in \Omega_q \mid \mu_k(z)=z\}=\{ \zeta_q^j \mid 1\le j\le k-1 \&\, %q | (k-1)j\}.$$
%For instance, if $q=6$ and $k=5$  we find $S_5=\{\pm 1\}.$
 %Let $D_k$ denote the group of diagonal matrices  with coefficients  in $S_k.$
%Call $M_{n,k}=<P_n,D_k>$ the group of matrices generated by the two groups  %$P_n$ and $D_k.$ Note that $$\forall M \in M_{n,k},\, \mu_k(M)=M. $$
Let $k$ be a fixed integer coprime to $q$. Define the {\bf strong group with multiplier} $\mu_{k}$, to be %the group of matrices in $M(n,q)$ satisfying the following property  strong 
$$ SAut(H,k)=\{M \in M(n,q) \mid \mu_k(M)H=HM\}.$$

{\prop The set of self-dual bent sequences for $H$ with multiplier $\mu_k$ is preserved globally under the action of $ SAut(H,k).$}

\begin{proof}

Let $X$ be such a sequence. Thus $HX=\lambda \mu_k(X)$ for some $\lambda \in \Z[\zeta_q].$ 
Multiplying both sides by $\mu_k(M)$ we see by using the relation $\mu_k(M)H=HM$ that $Y=MX$ is also a self-dual bent sequence for $H$ with multiplier $\mu_k.$ Indeed
$$ \mu_k(M)H X=HM X=\lambda \mu_k(M)\mu_k(X)=\lambda \mu_k(MX).$$
Observe that writing $\mu_k(M)\mu_k(X)=\mu_k(MX)$ uses the fact that $\mu_k$ can be extended to an element of the Galois group of $\Q(\zeta_q)/\Q.$
\end{proof}

A graph-theoretic  algorithm to construct $ SAut(H,k)$ in the case $k=1$ can be immediately extended from the $BH(n,4)$ case in \cite{ML}. Define a modified digraph $G'(H)$ as follows.
We connect every pair $r(s,x)$ and $c(s,x)$ for $s=1,\dots,n$ and $x$ ranging over $\Omega_q$ by a path a length two via an extra vertex $I_{s,x}$ say. The following proposition is an analogue of
\cite[Corollary 2]{ML}. %is stated \textcolor{blue}{below} \sout{without proof} . 
{Moreover, the similar result hold  for any $k$ with $1 \leq k\leq q-1$ by mean of  the digraph $G^k(H)=(V^k,E^k)$ where $V^k=V\cup \{I(s,x)\,\colon \, 1\leq s\leq n, x\in \Omega_q \}$ and $E^k=E\cup \{\big(r(s,y), I(s,x)\big),\, \big(I(s,x), c(s,x)\big)\colon\, 1\leq s\leq n, x\in \Omega_q, y=x^k\}$.} 

}%\marginpar{\tiny Does this definitely hold for the more general definition of strong automorphism group we now use? The proof of \cite[Corollary 2]{ML} is for when $SAut(H)$ is defined as I have above, but the definition of Proposition 3.2 is surely larger? In any case, I think this result needs to have a proof in this more general setting.\textcolor{blue}{Andrés: I think it is done.}}

{\prop \label{prop_Saut} If $k=1,$ then the strong automorphism group of $H$ is isomorphic to the automorphism group of $G'(H).$ {In general, $SAut(H,k)$ is isomorphic to the automorphism group of $G^k(H)$ for $1\leq k\leq q-1$.}}

\begin{proof}
    Given $H\in BH(n,q)$, $G(H)=(V,E)$ denotes its associated digraph. If $f\colon V\rightarrow V$ is an automorphism of $G(H)$, by definition $f$ is a 
   $1$-to-$1$ map such that  $(f(v_1),f(v_2))\in E$  iff $(v_1,v_2)\in E$. Looking at the out/in-degree of the vertices,  then $f$ has to stabilize the set of row/column vertices, respectively. Furthermore, $f$ maps cycles into cycles. Hence, the general form of $f$ is
   $$
   \begin{array}{c}
        f(r(t,x))=r(t',y_t\cdot x) \\
        f(c(s,x))=c(s',z_s\cdot x) 
   \end{array}$$
   with $y_t, z_s\in \Omega_q$ satisfying that $(f(v_1),f(v_2))\in E$  iff $(v_1,v_2)\in E$.
 Thus, for every automorphism $f$ we can define two matrices $P, Q\in M(n,q)$, 
    where $P=P_1D_1$ with $P_1$ the permutation matrix defined by $\pi_r(t)=t'$  and $D_1$ is the diagonal matrix with $t$-th entry equal to $\overline{y}_t$ (where $\overline{y}$ denotes the complex conjugation of $y$). Similarly, $Q=P_2D_2$ with $P_2$ the permutation matrix defined by $\pi_c(s)=s'$ and $D_2$ is the diagonal matrix with $s$-th entry equals to $\bar{z}_s$. Since $G(H)=G(PHQ^*)   $, it follows that $H=PHQ^*$. This proves Proposition \ref{prop_aut}.  
   
   {For $SAut(H,1)$ we consider the digraph $G'(H)=(V',E')$ defined above. If $f\colon V'\rightarrow V'$ an authomorphism of $G'(H)$ then,  by an argument on the out/in-degree of the vertices, $f$ also stabilizes the set of vertices $I(s,x)$. That is,
   $$f(I(s,x))=I(s',w_{s,x}\cdot 
   x)$$
   taking into account that $(f(r(s,x)),f(I(s,x)))$ and $(f(I(s,x)),f(c(s,x)))$ belong to $E'$, then
   we have $\pi_r(s)=\pi_c(s)$ and $y_s=w_{s,x}=z_s$. Therefore, $P=Q$ where $(P,Q)$ are the monomial matrices associated to $f$. This completes the proof of the first part.}

{In general, for $SAut(H,k)$ we consider the digraph $G^k(H)=(V^k,E^k)$. By a similar argument as for $k=1$, we conclude that 
$P=\mu_k(Q)$ where $PHQ^*=H$. As a consequence, $SAut(H,k)$ is isomorphic to the automorphism group of $G^k(H).$
     }  
\end{proof}

The automorphism group of $G'(H)$ can then be computed using standard graph theoretic programs like, e.g. Nauty \cite{N}. 

    \section{Computational  techniques}
    In this section we extend the computational construction methods of self-dual bent sequences of  \cite{ML,S+} to Butson matrices.
	\subsection{Groebner bases}
	The system $HX=\lambda\mu_k(X)$ with $X\in\Omega_q^n$ can be solved using Groebner bases, as it is equivalent to a polynomial system, which,
instead of being quadratic as in \cite{ML} is now of degree $q.$ %\marginpar{\tiny Isn't $k < q$ always?}
	More concretely, we can consider the following steps.
	\begin{enumerate}[(1)]
		\item Construct the ring $P$ of polynomial functions in $n$ variables $X_i,i=1,\dots,n.$
		\item Construct the linear constraints $HX=\lambda\mu_k(X).$
		\item Construct the constraints $\forall i\in \{1,2,\dots,n\}, X_i^q=1.$
		\item Compute a Groebner basis for the ideal $I$ of $P$ determined by constraints (ii) and (iii).
		\item Compute the solutions as the zeros determined by $I$.
	\end{enumerate}
	
{\bf Complexity:} As is well-known, the complexity of computing Groebner bases can be doubly exponential in the number of variables \cite{Grob}. However, for the system at hand, it might at most be simply exponential \cite{Go}.
	
	\subsection{Linear algebra}
	Consider the system $HX=\lambda\mu_k(X).$ Denote by $t$ the multiplicative order of $k$ modulo $q,$ implying $\mu_k^t=1.$ We can then obtain, by repeated application of $\mu_k$ the equations $$\mu_k^i(H)\mu_k^i(X)=\mu_k^i(\lambda)\mu_k^{i+1}(X), i=1,2,\dots,t-1.$$ By successive substitution, we get $$\prod_{i=0}^{t-1}\mu_k^{t-1-i}(H)X=\prod_{i=0}^{t-1}\mu_k^{t-1-i}(\lambda)X .$$
	From this equation, we can give an algorithm for the construction of self-dual bent sequences.
	
	\begin{enumerate}[(1)]
		\item Construct $H \in BH(n,q).$
		\item Compute $\prod_{i=0}^{t-1}\mu_k^{t-1-i}(H)$ and $\prod_{i=0}^{t-1}\mu_k^{t-1-i}(\lambda).$
		\item Compute a basis of the eigenspace associated to the eigenvalue $\prod_{i=0}^{t-1}\mu_k^{t-1-i}(\lambda)$ of $\prod_{i=0}^{t-1}\mu_k^{t-1-i}(H).$
		\item Let $B$ denote a matrix with rows such a basis of size $l\le n.$ Pick $B_l$ a $l$-by-$l$ submatrix of $B$ that is invertible, by the algorithm given below.
		\item For all $Z\in\Omega_q^l$ solve the system in $C$ given by $Z=CB_l.$
		\item Compute the remaining $n-l$ entries of $CB.$
		\item If these entries are in $\Omega_q$ declare $CB$ a self-dual bent sequence attached to $\prod_{i=0}^{t-1}\mu_k^{t-1-i}(H).$
		\item Sieve out the self-dual bent sequence that satisfies $HX=\lambda\mu_k(X).$
	\end{enumerate}
	
	To construct $B_l$ we apply a greedy algorithm. We construct the list $J$ of the indices of the columns of $B_l$ as follows.
	\begin{enumerate}[(i)]
		\item Initialize $J$ at $J=[1].$
		\item Given a column of index ${l}$ we compute the ranks $r$ and $r'$ of the submatrices of $B$ with $l$ rows and columns defined by the respective lists $J$ and $J'=Append(J,{l}).$
		\item If $r<r',$ then update $J:=J'.$
		\item Repeat until $\left | J \right |=rank(B). $
	\end{enumerate}

	{\re If the first column of $B$ is zero, step (i) does not make sense, but then there is no self-dual bent sequence in that situation, as all eigenvectors have first coordinate zero.\newline}

{\bf Complexity:} If we estimate at $O(n^3)$ the complexity of computation of an eigenspace basis, then the total complexity is $O(n^3+nkq^k)$ where $k$ is the dimension of the relevant eigenspace.
	
	\subsection{Numerical results}
	In Table 1, we apply the two previous construction techniques to the matrices of database \cite{W}. Each row corresponding to two values of $(n,q).$ Blanks mean that we cannot conclude, due to the computational burden.
	Explicit self-dual bent sequences can be seen at the bottom line of { \tt http://math.ahu.edu.cn/smj/list.htm}. In Table 2 we recall the distribution of matrices per classes of $ BH(n,q)$ in the said database. 
 \pagebreak
	\begin{longtable}{|c|c|c|c|c|c|c|c|}			
	\caption{Statistics of self-dual bent sequences of Butson matrices}\\
	\hline
	$n$ & $q$ & \#$\lambda$ & \#$X$ of per $\lambda$ &
	$n$ & $q$ & \#$\lambda$ & \#$X$ of per $\lambda$\\
	\hline
	2&2&0&0  &3&3&1&1\\   \hline
	4&2&2&2; 2   &4&4&3&8; 8; 4\\    \hline
	5&5&2&6; 10  &6&3&0&0\\   \hline
	6&4&0&0   &6&6&0&0 \\   \hline
	7&6&5&6; 4; 6; 0; 6   &7&7&2&9; 3  \\   \hline
	8&2&0&0  &8&4&0&0 \\   \hline
	8&4&0&0  &8&8&4&16; 16; 16; 16   \\   \hline
	8&20&0&0  &8&20&0&0   \\   \hline
	9&3&2&22; 12   &9&6&2&12; 4   \\   \hline
	9&9&2&48; 0   &9&10&2&2; 0   \\   \hline
	10&4&2&12; 12   &10&5&0&0   \\   \hline
	10&5&0&0  &10&6&0&0    \\   \hline
	10&6&0&0   &10&10&0&0   \\   \hline
	11&11&2&15; 5  &12&2&0&0    \\   \hline
	12&3&2&15; 1  &12&4&0&0   \\   \hline
	12&4&0&0   &12&6&2&96; 96   \\   \hline
	12&6&0&0   &12&12&4&72; 96; 72;108   \\  \hline
	12&36&3&   &13&6&0&0   \\ 	\hline
	13&6&0&0   &13&6&0&0   \\ 	\hline
	13&13&2&   &13&60&5&    \\  \hline
	14&4&0&0   &14&4&0&0     \\   \hline
	14&6&0&0   &14&7&0&0      \\   \hline
	14&10&0&0  &14&14&0&0      \\  \hline
	15&15&2&   &16&2&2&20; 20      \\  \hline
	16&2&2&8; 4   &16&2&2&0      \\   \hline
	16&2&2&0      &16&2&2&0      \\   \hline
	16&4&4&240; 240; 268; 180   &16&4&2&0   \\   \hline
	16&8&4&    &16&8&2&      \\   \hline
	16&16&4&    &21&3&0&0    \\   \hline
	\end{longtable}

	\begin{longtable}{|c|c|}
	\caption{Number of matrices in $BH(n,q)$ listed in the database \cite{W}}\\
	\hline
	\# & $BH(n,q)$\\
	\hline
     1 & \thead{$BH(2,2), BH(3,3), BH(4,2), BH(4,4), BH(5,5), BH(6,3), BH(6,4),BH(6,6), BH(7,6),$\\$BH(7,7), BH(8,2), BH(8,8), BH(9,3), BH(9,6), BH(9,9), BH(9,10), BH(10,4), BH(10,10),$\\$ BH(11,11), BH(12,2), BH(12,3), BH(12,12), BH(12,36), BH(13,13), BH(13,60)$\\} \\
     \hline
     2 & \thead{$BH(8,4), BH(8,20), BH(10,5), BH(10,6), BH(12,4), BH(12,6), BH(14,4), BH(16,4),$\\$ BH(16,8)$}\\
     \hline
     3 & $BH(13,6)$ \\
     \hline
     5 & $BH(16,2)$\\
     \hline
		\end{longtable} %\marginpar{\tiny We should be clear that this table is not reporting on a complete classification of Butson matrices at these orders, it's just what's available from a particular database.THAT was said BEFORE}
			
\section{General constructions}
\subsection{Kronecker product}
If we keep the field of definition and the multiplier then we can take the Kronecker product of Butson matrices with self-dual bent sequences and
obtain another Butson matrix with a self-dual bent sequence, as the next result shows.

{\prop Let $H$ (resp. $K$) be a Butson matrix in $BH(n,q)$ (resp. $BH(m,q)$) affording a self dual bent sequence $X$ (resp. $Y$) for the multiplier $\mu_k.$
Then $X\otimes Y$ is a self-dual bent sequence with multiplier $\mu_k,$ for  $H\otimes K \in BH(mn ,q).$ }

\begin{proof} From $HX=\lambda \mu_k(X)$ and $KY=\nu \mu_k(Y)$ we obtain, taking Kronecker products $$H\otimes K(X\otimes Y)= \lambda \nu \mu_k(X)\otimes \mu_k(Y).$$ The result follows by $\mu_k(X)\otimes \mu_k(Y)=\mu_k(X\otimes Y).$
\end{proof}

\noindent{\bf Example:} Taking $q=3,$ we write $j$ for a complex cubic root of unity. We find an example of application of this proposition for $n=3,\,m=9,\ k=2$ when
$X=(1,j,j),$ and $Y=(1,1,1,1,j^2,j,1,j,j^2).$
\subsection{Regular matrices}
A Butson Hadamard matrix $H$ of order $n$ is $regular$ if it has constant row and column sum. A direct connection between self-dual bent sequences and regular Butson Hadamard matrices is as follows. Denote by ${\bf 1}_n$ the all-one vector of length $n.$

{\prop If $H$ is a regular Butson Hadamard matrix of order $n,$ then for all $ u \in \Omega_q,$ we have that $u{\bf 1}_n$ is a self-dual bent sequence for $H$ with trivial multiplier.}

\begin{proof}
	Denote by $\sigma$ the sum of elements of any row. By definition of regular Butson Hadamard matrices $H{\bf 1}_n=\sigma {\bf 1}_n.$
The result follows by linearity.
\end{proof}

A strong existence condition for regular Butson matrices is Lemma 5.2 in \cite{E+}.

\subsection{Bush type}
A Bush-type Hadamard matrix, say $H,$ is a Butson Hadamard matrix of order $n^2$ over the $q^{th}$ roots of unity which is subdivided into $n^2$ blocks $H_{11},H_{12},\dots,H_{nn}$ of order $n$ such that $JH_{ij}=H_{ij}J=\delta_{ij}nJ,$ where $J$ denotes the all ones matrices.
Bush type Butson matrices were constructed recently in \cite{K+}. They have many self-dual bent sequences attached to $H$ as the next result shows.

{\prop If $H$ is a Bush-type Hadamard matrix of order $n^2$ over the $q^{th}$ roots of unity, then there are at least $q^n$ self-dual bent sequences with trivial multiplier attached to $H.$}

\begin{proof}
	From the definition, we see that the sequence $X$ defined by $X^T=(u_1 {\bf 1}_n,\dots,u_n{\bf 1}_n),$ where the $u_k$'s are arbitrary in $\Omega_q,$ is a self-dual bent sequence.
\end{proof}

\subsection{Fourier transforms}
Let $q \ge 2$ be an integer, not necessarily a prime power. Let $r \ge 1$ be an integer. We define the inner product $\cdot$ on $\Z_q^r$ as
\be \label{e0-FO}
(x_1, \ldots, x_r) \cdot (y_1, \ldots, y_r)=x_1y_1 + \cdots + x_ry_r,
\ee
where we use the ring structure of $\Z_q$ on the right-hand side. Note that $x\cdot y=y \cdot x$ for all $x,y \in \Z_q^r$.

Let $n=q^r$ and $H$ be the $n \times n$ matrix over $\C$ defined as follows: For $x,y \in \Z_q^r$,
the entry $H(x,y)$ of $H$ corresponding to row $x$ and  column $y$ is given by
\be \label{e1-FO}
H(x,y)=\zeta_q^{x \cdot y}.
\ee
It follows from the orthogonality of group characters that $H$ is a Butson type Hadamard matrix of order $n$ belonging to $BH(n,q).$

Let $f:\Z_q^r \ra \Z_q$ be a function. Let $X=X(f)$ be the $r \times 1$ matrix over $\C$ such that the $x$-th entry of $X$ is $\zeta_q^{f(x)}$. We assume that $f$ is not the zero function so that $X$ is not the trivial column matrix, i.e. not the full one column.
Recall that if there exists a multiplier $\mu_k: \Omega_q \ra \Omega_q$ and $\lambda \in \Z[\zeta_q]$ such that $HX=\lambda \mu_k(X)$, then $X$ is called a self-dual bent sequence attached to $H$ with multiplier $\mu_k$ and the coefficient $\lambda$ (see Definition \ref{definition.self-dual.bent} above).
{
Let us point out that
when $X$ is a self-dual bent sequence attached to  the conjugate matrix of  (\ref{e1-FO}), then $f: \Z_q^r \ra \Z_q$ is a generalized bent function (in sense of \cite{Sch19}).
}

In the following theorem, when $r$ is even and $f$ is a generalized Maiorana-McFarland (MM) function, we characterize all $f$ such that $X(f)$ is self-dual.

\begin{thm} \label{thm1.FO}
Let $q \ge 2$ be an integer, not necessarily a prime power. Let $m \ge 1$ be an integer. Let $r=2m$. Let $\phi: \Z_q^m \ra \Z_q^m$ be a permutation. Let $f: \Z_q^r \ra \Z_q$ be the map defined as
\be
f(x_1,x_2)=x_1 \cdot \phi(x_2),
\nn\ee
where $\cdot$ is the Euclidean inner product on $\Z_q^m$ defined as in (\ref{e0-FO}). 

Let $1 \le k \le q-1$ and $\mu_k(z)=z^k$ be as in Definition \ref{definition.self-dual.bent}. Put $n=q^r$
and let $H$ be the matrix of order $n$ given in (\ref{e1-FO}). 
Let $X$ be the $n \times 1$ column such that $(x_1,x_2)$-th entry of $X$ is $\zeta_q^{f(x_1,x_2)}$. Then $X$ is a self-dual bent sequence attached to $H$ with
\be
HX=\lambda \mu_k(X)
\nn\ee
if and only if all of the followings hold:
\begin{itemize}
\item[(1)] $x_1 \cdot x_2 + k \phi(x_1) \cdot \phi(x_2)=0$ for all $x_1,x_2 \in \Z_q^m$.
\item[(2)] $\lambda=q^m$.
\end{itemize}
\end{thm}
\begin{proof}
Let $(x_1,x_2) \in \Z_q^r$. The $(x_1,x_2)$-th entry of $HX$ is
\be
\begin{array}{l}
\dd \sum_{y_1 \in \Z_q^m} \sum_{y_2 \in \Z_q^m} \zeta_q^{(x_1,x_2) \cdot (y_1,y_2)} \zeta_q^{f(y_1,y_2)} =\sum_{y_1 \in \Z_q^m} \sum_{y_2 \in \Z_q^m} \zeta_q^{x_1 \cdot y_1 + x_2\cdot y_2 + y_1 \cdot \phi(y_2) } \\ \\
 \dd =\sum_{y_2 \in \Z_q^m} \zeta_q^{ x_2 \cdot y_2 } \sum_{y_1 \in \Z_q^m} \zeta_q^{ y_1 \cdot \left( x_1+\phi(y_2)\right)}  =q^m \zeta_q^{ \phi^{-1}(-x_1) \cdot x_2 }.
\end{array}
\nn\ee
Hence $X$ is a self-dual bent sequence if and only if $\lambda=q^m$ and
\be
\phi^{-1}(-x_1) \cdot x_2 = k x_1 \cdot \phi(x_2) \;\; \mbox{for all $x_1,x_2 \in \Z_q^m$}.
\nn\ee
Consider the change of variable $x_1'=\phi^{-1}(-x_1) \iff \phi(x_1')=-x_1$. Using this change of variable, the last condition is equivalent to the condition that $x_1 \cdot x_2 + k \phi(x_1) \cdot \phi(x_2)=0 \;\; \mbox{for all $x_1,x_2 \in \Z_q^m$}.$
\end{proof}

The following corollary is immediate.

\begin{cor} \label{cor1.FO}
  Let $q \ge 2$ be an integer, not necessarily a prime power. Let $m \ge 1$ be an integer. Let $r=2m$.  Let $H$ be the matrix of order $n$ given in (\ref{e1-FO}). Consider the set
  \be
  S=\left \{ \frac{-1}{d^2}: 1\le d \le q-1, \; \; \gcd(d,q)=1 \right \}. 
  \nn\ee
If $k \in S$, then there exists a self-dual bent sequence $X$ attached to $H$ with $HX=q^m\mu_k(X)$.
\end{cor}
\begin{proof}
Let $1\le d \le q-1$ be an integer with $\gcd(d,q)=1$. Note that the map $\phi(x)=dx$ is a permutation on $\Z_q^m$. Moreover we have
\be
x_1 \cdot x_2 + k \phi(x_1) \cdot \phi(x_2)= (1+kd^2) x_1\cdot x_2
\nn\ee
for all $x_1,x_2 \in \Z_q^m$. If $k \in S$, then $1+kd^2=0$, which completes the proof.
\end{proof}
{
\begin{re}\label{rm1}
    Taking $d=1$ in Corollary \ref{cor1.FO}, then for $k=q-1$ and $\phi(x)=x$, we have that
    $g(x_1,x_2)=x_1\cdot x_2$ is a (self-dual) generalized bent function with $HX= q^m \;\overline{X}$.
\end{re}
}

Next, we consider some group invariant Butson type Hadamard matrices having self-dual bent sequences.

Let $q \ge 2$ be an integer, not necessarily a prime power. Let $m \ge 1$ be an integer and $r=2m$. Let $n=q^r$ and $H$ be the $n \times n$ matrix over $\C$ defined as follows: For $(x_1,x_2), (y_1,y_2) \in \Z_q^r$,
the entry $H((x_1,x_2),(y_1,y_2))$ of $H$ corresponding to the row $(x_1,x_2)$ and the column $(y_1,y_2)$ is given by
\be \label{e2-FO}
H((x_1,x_2),(y_1,y_2))=\zeta_q^{(x_1-y_1) \cdot (x_2-y_2)}{=\zeta_q^{g((x_1,x_2)-(y_1,y_2))}
\quad \mbox{for $g$ from Remark \ref{rm1}}}.
\ee
This implies that
\be
H((x_1+a_1,x_2+a_2),(y_1+a_1,y_2+a_2))=H((x_1,x_2),(y_1,y_2))
\nn\ee
for all $(a_1,a_2) \in Z_q^r$ and hence the matrix $H$
defined in (\ref{e2-FO}) is $\Z_q^r$-group invariant in the sense of \cite{TD}. Note that the matrix defined in (\ref{e1-FO}) is not  $\Z_q^r$-group invariant. {Let us point out that $H(x,y)=\zeta_q^{g(x-y)}$ is a group invariant Butson type Hadamard matrix when $g\colon \Z_q^n \ra \Z_q$  is a generalized bent function (see \cite[Proposition 2.3]{Sch19}).}

Now we are ready to present the analogous result of Theorem \ref{thm1.FO} for the group invariant
matrix $H$ defined in (\ref{e2-FO}). We also show that $H$ is a Butson type Hadamard matrix in the next theorem.

\begin{thm} \label{thm2.FO}
Let $q \ge 2$ be an integer, not necessarily a prime power. Let $m \ge 1$ be an integer. Let $r=2m$. Let $\phi: \Z_q^m \ra \Z_q^m$ be a permutation. Let $f: \Z_q^r \ra \Z_q$ be the map, different from Theorem \ref{thm1.FO},  defined as
\be
f(x_1,x_2)=x_1 \cdot \phi(x_2) - x_1 \cdot x_2,
\nn\ee
where $\cdot$ is the Euclidean inner product on $\Z_q^m$ defined as in (\ref{e0-FO}). 
Assume that $f$ is not the zero map.

Let $1 \le k \le q-1$ and $\mu_k(z)=z^k$ be as in Definition \ref{definition.self-dual.bent}. Put $n=q^r$
and let $H$ be the matrix of order $n$ given in (\ref{e2-FO}). 
We have that $H$ is a Butson type Hadamard matrix.

Moreover, let $X$ be the $n \times 1$ column such that $(x_1,x_2)$-th entry of $X$ is $\zeta_q^{f(x_1,x_2)}$. Then $X$ is a self-dual bent sequence attached to $H$ with
\be
HX=\lambda \mu_k(X)
\nn\ee
if and only if all of the followings hold:
\begin{itemize}
\item[(1)] $k x_1 \cdot \phi^2(x_2) - (k+1) x_1 \cdot \phi(x_2) + x_1 \cdot x_2=0$ for all $x_1,x_2 \in \Z_q^m$.
\item[(2)] $\lambda=q^m$.
\end{itemize}
\end{thm}
\begin{proof}
We first prove that $H$ defined in (\ref{e2-FO}) is a Butson type Hadamard matrix. Let $(x_1,x_2),(z_1,z_2) \in \Z_q^r$. The entry of $HH^*$ corresponding to the row $(x_1,x_2)$ and the column $(z_1,z_2)$ is
\be
\begin{array}{l}
\dd \sum_{y_1 \in \Z_q^m} \sum_{y_2 \in \Z_q^m} \zeta_q^{(x_1-y_1)\cdot (x_2-y_2)} \zeta_q^{-(y_1-z_1)\cdot (y_2-z_2)} =\sum_{y_1 \in \Z_q^m} \sum_{y_2 \in \Z_q^m} \zeta_q^{x_1 \cdot x_2 - z_1\cdot z_2 + y_1 \cdot (z_2-x_2) + y_2 \cdot (z_1-x_1) } \\ \\
 \dd =\zeta_q^{x_1 \cdot x_2 - z_1\cdot z_2}  \sum_{y_1 \in \Z_q^m} \zeta_q^{ y_1 \cdot (z_2-x_2) } \sum_{y_2 \in \Z_q^m} \zeta_q^{ y_2 \cdot (z_1-x_1) }.
\end{array}
\nn\ee
These arguments and the usual orthogonality relations imply that $HH^*=q^rI_n$. This completes the proof of the statement that $H$ defined in (\ref{e2-FO}) is a Butson type Hadamard matrix.

Let $(x_1,x_2) \in \Z_q^r$. The $(x_1,x_2)$-th entry of $HX$ is
\be
\begin{array}{l}
\dd \sum_{y_1 \in \Z_q^m} \sum_{y_2 \in \Z_q^m} \zeta_q^{(x_1-y_1) \cdot (x_2-y_2)} \zeta_q^{f(y_1,y_2)} =\sum_{y_1 \in \Z_q^m} \sum_{y_2 \in \Z_q^m} \zeta_q^{x_1 \cdot x_2 + y_1\cdot y_2 -x_1\cdot y_2 -x_2 \cdot y_1 -y_1 \cdot y_2 + y_1 \cdot \phi(y_2) } \\ \\
 \dd =\zeta_q^{x_1 \cdot x_2} \sum_{y_2 \in \Z_q^m} \zeta_q^{ -x_1 \cdot y_2 } \sum_{y_1 \in \Z_q^m} \zeta_q^{ y_1 \cdot \left(\phi(y_2)-x_2\right)}  =q^m \zeta_q^{ x_1\cdot x_2 -x_1\cdot \phi^{-1}(x_2)}.
\end{array}
\nn\ee
Hence $X$ is a self-dual bent sequence if and only if $\lambda=q^m$ and
\be
x_1 \cdot x_2 - x_1 \cdot \phi^{-1}(x_2) = -kx_1\cdot x_2 + k x_1 \cdot \phi(x_2) \;\; \mbox{for all $x_1,x_2 \in \Z_q^m$}.
\nn\ee
By a suitable change of variables as in the proof of Theorem \ref{thm1.FO}, the last condition is equivalent to the condition that $kx_1 \cdot \phi^2(x_2) - (k+1)x_1 \cdot \phi(x_2) + x_1 \cdot x_2=0 \;\; \mbox{for all $x_1,x_2\in \Z_q^m$}.$
\end{proof}

The following corollary is immediate.

\begin{cor} \label{cor2.FO}
  Let $q \ge 3$ be an integer, not necessarily a prime power. Let $m \ge 1$ be an integer. Let $r=2m$.  Let $H$ be the matrix of order $n$ given in (\ref{e2-FO}). Consider the set
  \be
  S=\left \{ \frac{1}{d}: 1\le d \le q-1, \; \; \gcd(d,q)=1,\; \; d \neq 1 \right \} 
  \nn\ee
If $k \in S$, then there exists a self-dual bent sequence $X$ attached to $H$ with $HX=q^m\mu_k(X)$.
\end{cor}
\begin{proof}
Let $1\le d \le q-1$ be an integer with $d \neq 1$. Note that the map $\phi(x)=dx$ is a permutation on $\Z_q^m$. Also $f$ is not the zero map as $d \neq 1$. Moreover, we have
\be
k x_1 \cdot \phi^2(x_2) - (k+1) x_1 \cdot \phi(x_2) + x_1 \cdot x_2= \left(kd^2-(k+1)d+1\right) x_1\cdot x_2
=(kd-1)(d-1) x_1 \cdot x_2
\nn\ee
for all $x_1,x_2 \in \Z_q^m$. If $k \in S$, then $kd-1=0$, which completes the proof.
\end{proof}

\begin{re}
Corollaries \ref{cor1.FO} and \ref{cor2.FO} only present very simple examples satisfying conditions of Theorems \ref{thm1.FO} and \ref{thm2.FO}, respectively. For example, if $q=2$ and $m \ge 1$ is any integer, there are also such maps $\phi$ proving self-dual bent sequences satisfying the conditions of Theorem \ref{thm2.FO}.
\end{re}

	\section{Existence conditions}
The following result is inspired by \cite[Lemma 5.2]{E+}. The hypothesis holds true when $H$ admits a self-dual bent sequence with $k=1.$
	{\thm \label{excl} If $H$  in $BH(n,q)$  admits an eigenvalue $\lambda$ in $\Z[\zeta_q],$ with an eigenvector in $\Omega_q^n,$ then there are $q$ nonnegative integers $y_r\le n$ summing up to $n,$ and such that $$ n=\left ( \sum_{r=0}^{q-1} y_rc_r \right )^2+\left ( \sum_{r=0}^{q-1} y_rs_r \right )^2 $$
		where $\zeta_q^r=c_r+is_r$ and $i=\zeta_4.$}
	
	\begin{proof}
		Projecting the eigenvalue equation $Hx=\lambda x$ on the first coordinate yields after dividing out by $x_1$ the relation
		$$\sum_{j=1}^n H_{1j} \frac{x_j}{x_1}=\lambda.$$
		Let $y_r$ denote the number of occurrences of $\zeta_q^r$ among the summands in the above sum. By definition
		$$\sum_{r=0}^{q-1} y_r=n , \;\mbox{where $0\le y_r \le n$}.$$
		A well-known consequence of the Hadamard equation is that
		$$n=|\lambda|^2=a^2+b^2,$$
		where $a=\Re(\lambda),\, b=\Im(\lambda).$ The result follows upon computing $a$ and $b$ from the above expression for $\lambda.$
	\end{proof}
	
	{\re In the literature of combinatorics, the sequence $(y_r)_{r=0}^{q-1}$ is called a {\bf composition} of $n$ into $q$ parts. }
	
	{\ex $H\in BH(6,3)$, $n=6, q=3.$ There are precisely $\binom{6+3-1}{6}=\binom{8}{6}=28$ compositions of $6$ into $3$ parts, namely,
		$$ \begin{aligned}
			6= 6+0+0 = 0+6+0 = 0+0+6 = 1+2+3 = 5+1+0 = 5+0+1 = 1+5+0 \\
			= 2+1+3 = 1+0+5 = 0+1+5 = 0+5+1 = 2+2+2 = 4+2+0 =4+0+2 \\
			=0+4+2 =2+4+0 =2+0+4 =0+2+4 =4+1+1 =1+4+1 =1+1+4 \\
			=3+3+0 =3+0+3 =0+3+3 =3+2+1 =3+1+2 =1+3+2 =2+3+1
		\end{aligned}  $$
		$6=6+0+0=0+6+0=0+0+6,$ when $a^2+b^2=36.$\\
		$6=5+1+0=5+0+1=1+5+0=1+0+5=0+1+5=0+5+1,$ when $a^2+b^2=21.$\\
		$6=4+2+0=4+0+2=0+4+2=2+0+4=0+2+4=2+4+0,$ when $a^2+b^2=12.$\\
		$6=4+1+1=1+1+4=1+4+1=3+3+0=3+0+3=0+3+3,$ when $a^2+b^2=9.$\\
		$6=3+2+1=3+1+2=1+3+2=2+3+1=1+2+3=2+1+3,$ when $a^2+b^2=3.$\\
		$6=2+2+2,$ when $a^2+b^2=0.$\\
		The above $28$ cases satisfy $\sum_{r=0}^{q-1}y_r=6,$ but $a^2+b^2\ne 6,$ so $BH(6,3)$ does not contain matrices having eigenvalues of the type of Theorem \ref{excl}. This is confirmed by computation on matrices of $BH(6,3).$
	}
	{\re The condition in Theorem \ref{excl} is necessary but not sufficient as $BH(8,4)$ contains two matrices neither leading to self-dual bent sequences. Still, $8=2^2+2^2,$ with $y_0=4,\,y_1=y_2=2,\,y_3=0.$ A similar phenomenon occurs for $BH(21,3).$}
	
	\pagebreak
		{\begin{longtable}{|c|c|c|c|}
	        	\caption{Parameters excluded by Theorem \ref{excl}}\\
				\hline
				n & q & \# of compositions & values of $a^2+b^2$ \\
				\hline
				2 & 2 & 3 & 0;4 \\
				\hline
				6 & 3 & 28 & 0;3;9;12;21;36 \\
				\hline
				6 & 4 & 84 & 0;2;4;8;10;16;18;20;26;36 \\
				\hline
				6 & 6 & 462 & 0;1;3;4;7;9;12;13;16;19;21;25;27;28;31;36 \\
				\hline
				8 & 2 & 9 & 0;4;16;36;64 \\
				\hline
				8 & 4 & 165 & 0;2;4;8;10;16;18;20;26;32;34;36;40;50;64 \\
				\hline
				10 & 6 & 3003 & \thead{0;1;3;4;7;9;12;13;16;19;21;25;27;28;31;36;37;39;43;48;49;52;
					57;61;63;64;\\67;73;75;76;79;81;84;91;100} \\
				\hline
				12 & 2 & 13 & 0;4;16;36;64;100;144 \\
				\hline
				12 & 4 & 455 & \thead{0;2;4;8;10;16;18;20;26;32;34;36;40;50;52;58;64;68;72;74;80;82;90;100;104;\\122;144} \\
				\hline
				12 & 6 & 6188 & \thead{0;1;3;4;7;9;12;13;16;19;21;25;27;28;31;36;37;39;43;48;49;52;57;61;63;64;\\67;73;75;76;79;81;84;91;93;97;100;103;108;109;111;112;117;121;124;133;144} \\
				\hline
				13 & 6 & 8568 & \thead{0;1;3;4;7;9;12;13;16;19;21;25;27;28;31;36;37;39;43;48;49;52;57;61;63;64;\\67;73;75;76;79;81;84;91;93;97;100;103;108;109;111;112;117;121;124;127;129;133;139;\\144;147;157;169 }\\
				\hline
				14 & 4 & 680 & \thead{0;2;4;8;10;16;18;20;26;32;34;36;40;50;52;58;64;68;72;74;80;82;90;98;100;104;\\106;116;122;130;144;148;170;196} \\
				\hline
				14 & 6 & 11628 & \thead{0;1;3;4;7;9;12;13;16;19;21;25;27;28;31;36;37;39;43;48;49;52;57;61;63;64;\\67;73;75;76;79;81;84;91;93;97;100;103;108;109;111;112;117;121;124;127;129;133;139;\\144;147;148;151;156;157;163;169;172;183;196} \\
				\hline
				21 & 3 & 253 & \thead{0;3;9;12;21;27;36;39;48;57;63;75;81;84;93;108;111;117;129;144;147;156;171;183;\\189;201;225;228;237;273;279;324;327;381;441} \\
				\hline
				\end{longtable} }

	\section{Distance spectrum of Butson Hadamard codes}
	{\de The {\bf unit sphere} $\Upsilon_d$ in Euclidean $d$-space $\mathbb{R}^d$ is the set of all unit norm vectors:$$\Upsilon_d\triangleq\left \{ x=(x_1,x_2,\dots,x_d)\in \mathbb{R}^d:\left \| x \right \|=1 \right \} $$}
{\de The {\bf hermitian inner product} of $x$ and $y$ in $\mathbb{C}^n$ is defined as
$$\left \langle x,y \right \rangle=\sum_{i=1}^n x_i \overline{y_i},$$
with $\overline{z}$ denoting the complex conjugate of $z.$ The {\bf (squared) Euclidean distance} of $x$ and $y$ is then
$$d_E(x,y)=\left \langle x-y,x-y \right \rangle=\left \langle x,x \right \rangle+\left \langle y,y \right \rangle-2 \Re(\left \langle x,y \right \rangle).$$}
	{\de A {\bf spherical code} in dimension $d$ is a finite set $X \subseteq \Upsilon_d.$ Its minimum distance for the squared Euclidean distance is denoted by $\rho.$
Its parameters are denoted compactly by $(d,\rho,|X|).$ The function $A_d(\rho)$ can then be defined as
$$A_d(\rho)=\max\{ |X| \mid X \text{ spherical code of parameters } (d,\rho,|X|)\}.$$
}
In the next proposition we compute the possible Chinese Euclidean distances of ${C_H}$, which are also the squared Euclidean distances of the spherical code $\phi(\zeta_q^{C_H})$ to be defined below.
	{\prop %For simplicity put $w=\zeta_q$ and write $w^{F_H+i}$ for $w^i w^{F_H}.$ If $w^{C_H}=w^{F_H}\cup w^{F_H+1}\cup \cdots \cup w^{F_H+q-1}$
 The Chinese Euclidean distances of ${C_H}$  are
		$$\left \{ d_E(x,y) \mid x\ne y, x,y\in w^{C_H}\right \}=
		\left \{ 2n \right \} \cup \left \{ 2n(1-\cos\frac{2\pi t}{q}) \mid t=1,2,\dots,\left \lfloor \frac{q}{2} \right \rfloor \right \}. $$ 	}
 \begin{proof}
 For simplicity put $w=\zeta_q$ and write $w^{F_H+i}$ for $w^i w^{F_H}.$ Note that, by definition of $C_H$ we have  $w^{C_H}=w^{F_H}\cup w^{F_H+1}\cup \cdots \cup w^{F_H+q-1}.$
 The Chinese Euclidean distances of ${C_H}$ can then be computed as follows.
 	Let $x=w^{i}h, y=w^{j}h', $ where $h,h^{'}\in w^{F_H}$ and $i,j\in \left \{ 0,1,\dots,q-1 \right \}.$ We have
 	\[
 	\left \langle x,y \right \rangle = \left\{
 	\begin{array}{ll}
 		0, &h\ne h^{'},\\
 		nw^t, &h=h^{'},
 		\end{array}
 	\right.
 	\]
 	where $t\in \left \{ 1,2,\dots,q-1 \right \}.$ From the equation $d_E(x,y)=2n-2\Re(\left \langle x,y \right \rangle ),$ the result follows.
 \end{proof}

The map $\psi:\mathbb{C}\longrightarrow \mathbb{R}^2, x+iy\longmapsto(x,y) $ is an isometry from $\left ( \mathbb{C}^n,d_E \right ) $ to $\left ( \mathbb{R}^{2n},\delta \right ), $ where $\delta (U,V)=\left \| U-V \right \|^2,$ for all $U,V \in \mathbb{R}^{2n}.$ Note that
$(\psi(a),\psi(b))=\Re\big(\left \langle a,b \right \rangle\big), $
where $a$ and $b\in \mathbb{C}^n,$ and $ (,)$ denotes the standard
inner product in $\mathbb{R}^{2n}.$
For normalization purposes,
we will let $\phi(z)=\frac{\psi(z)}{\sqrt{n}},$ for all $z \in \mathbb{C}^n.$

{\cor \label{sphere} The spherical code $\phi(\zeta_q^{C_H}) \subseteq \Upsilon_{2n} \subseteq \mathbb{R}^{2n}$ has a size of $nq$ and a distance of $$\rho=\frac{d_{CE}}{n}=2(1-\cos\frac{2\pi}{q}).$$}

We note an unexpected consequence for the parameters of a Butson Hadamard matrix.

{\prop If $H\in BH(n,q)$ exits, then the following inequality holds.
	$$nq\le A_{2n}(2(1-\cos{\frac{2\pi}{q}})).$$}
We will require a special case of Levenshtein bound  \cite[Theorem 2.5.1]{EZ}.
	{\thm \label{lev} Let $X$ be a spherical code with parameters $(d,\rho,|X|)$ and let $s=1-\frac{\rho}{2}. $ Then the size $|X|$ is necessarily bounded according to $|X|\le L_3(s)=\frac{d(2+(d+1)s)(1-s)}{1-ds^2},$ for $s\in [0,\frac{1}{(\sqrt{d+3}+1)}).$ }

We deduce from Corollary \ref{sphere} the following optimality result in relation to complex Hadamard matrices in the Turyn sense.

{\cor If $H \in BH(n,4),$ then the spherical code $\phi(\zeta_q^{C_H})$ is optimal in dimension $2n.$}

\begin{proof}
By Corollary \ref{sphere} we know that, in the notation of Theorem \ref{lev}, we have $s=\cos(\frac{2\pi}{4})=0,$ and $d=2n.$ The result follows by
Theorem \ref{lev} since the size of the code is $4n=2d.$
\end{proof}

	\section{Covering radius of Butson Hadamard codes}
\subsection{Lower bound}
	{\de The {\bf deviation} of an arbitrary vector $x\in\Omega_q^n$ from a polyphase code $C$ is defined as $$\theta (C,x)=\max_{y\in C}\{\left | \left \langle x,y \right \rangle \right |\}, $$ where $\left \langle x,y \right \rangle$ denotes the hermitian inner product of $x$ and $y.$
}
	
	{\prop \label{dev} If $X$ is a bent sequence for $H\in BH(n,q),$ then its corresponding polyphase code $\zeta_q^{C_H}$ has deviation $\theta (C,X)=\sqrt{n},$ where $X\in\Omega_q^n.$}
	
	\begin{proof}
		$HX=\lambda \mu_k(X),$ where $k$ and $q$ are coprime.
		Let $h_i$ be the $i$-th row of $H,$ then $$\left | \left \langle h_i,X \right \rangle \right |=\left | \lambda \right |\left | \mu_k(X_i) \right |=\sqrt{n}.$$
		Since $\left | \left \langle \zeta_q^t h_i,X \right \rangle \right |=\left | \left \langle \ h_i,X \right \rangle \right |$ for all $i=1,\cdots,n$ and all $t=1,\cdots,q$ the result follows.
		\end{proof}

	%Let $C\subseteq \Omega_q^n,$ be a polyphase code for the squared Euclidean distance $d_E. $ Let $Z$ be the $\mathbb{Z}_q$-code determined by $\xi_q^Z=C$ for the Chinese Euclidean distance $d_{CE}.$\\
	We consider the polyphase codes attached to some $H\in BH(n,q)$ when $q$ equals $4,6$ and $8.$ By induction on $n,$ we can get $d_{CE}(u,v)=d_E(x,y),$ where $x=\zeta_q^u$ and $y=\zeta_q^v$ with $u,v\in \mathbb{Z}_q^n.$\\
	It can be seen by expanding $\left \langle x-y,x-y \right \rangle $ that $$Re(\left \langle x,y \right \rangle )=\frac{2n-d_E(x,y)}{2}=\frac{2n-d_{CE}(u,v)}{2},$$ for all $x,y\in \Omega_q^n.$\\
	The simple inequality $Re(\left \langle x,y \right \rangle )\le \left | \left \langle x,y \right \rangle  \right | $ in the complex field shows that
	$$r_{CE}(C_H)\ge 2n-2\theta(C,x).$$\\
	Combining this fact with the above proposition \ref{dev} yields the following bound.
	{\cor \label{low} If there is a bent sequence $X$ for $H\in BH(n,q),$ then the covering radius of its attached $\mathbb{Z}_q$-code $C_H$ is bounded below as
		$$r_{CE}(C_H)\ge 2n-2\sqrt n.$$}
	
	The following three tables indicate the lower bounds and the true values of the covering radius when $q$ equals $4, 6$ and $8,$ respectively, under the existence of the Butson Hadamard matrix with self-dual bent sequences. The second column  displays the index $k$ of the multiplier $\mu_k.$ A question mark indicates that the computation could not be completed.
		
	{\begin{longtable}{|c|c|c|c|}
			\caption{The covering radius of $C_H$ for $H\in BH(n,4)$}\\
			\hline
			$BH(n,4)$&$k$&Lower bound&True value\\
			\hline
			$BH(4,4)$&1;3&4&4\\
			\hline
			$BH(10,4)$&3&14&14\\
			\hline
			$BH(16,4)(F_4\otimes F_4)$&1;3&$24$& ?\\
			\hline
		\end{longtable}
		
		{\begin{longtable}{|c|c|c|c|}
			\caption{The covering radius of $C_H$ for $H\in BH(n,6)$}\\
				\hline
				$BH(n,6)$&$k$&Lower bound&True value\\
				\hline
				$BH(7,6)$&5&9&9\\
				\hline
				$BH(9,6)$&1;5&12&12\\
				\hline
				$BH(12,6)(F_6\otimes F_2)$&5&$2(12-\sqrt{12})\approx17.072$& ?\\
				\hline
			\end{longtable}	
			
			{\begin{longtable}{|c|c|c|c|}
				\caption{The covering radius of $C_H$ for $H\in BH(n,8)$}\\
				\hline
				$BH(n,8)$&$k$&Lower bound&True value\\
				\hline
				$BH(8,8)$&7&$2(8-\sqrt{8})\approx10.343$&$2(8-\sqrt{8})$\\
				\hline			
				\end{longtable}
\subsection{Upper bounds}
In this subsection, we derive upper bounds on the covering radius of the spherical code $\phi(\zeta_q^{C_H})$ (after normalization to belong to the unit sphere) defined in the preceding section, as a function of its strength as a spherical design.
This bound is then directly an upper bound on the covering radius of the code $C_H$ for the Chinese Euclidean distance. However, the two quantities may or may not coincide.
				{\de A  spherical code $X$ is a {\bf 1-design} if its center of mass is the origin or, more concretely, for all coordinate indices $i$ satisfy $ \sum_{x \in X} x_i=0.$ It is {\bf antipodal} if $X=-X.$}

{\prop If $H\in BH(n,q)$ with $q$ is even, then the covering radius of $\phi(\zeta_q^{C_H})$ is at most $\sqrt{2}.$ }

\begin{proof}
The spherical code $\phi(\zeta_q^{C_H})$ is  antipodal, because $-1$ is a power of $\zeta_q.$ The result follows then by \cite[Theorem
 1]{S}.
\end{proof}
 The same bound can be obtained from different hypotheses on $H.$

{\prop If $H$ is dephased, then $\phi(\zeta_q^{C_H})$ is a {\bf 1-design} and its covering radius is at most $\sqrt{2}.$ }
\begin{proof}
Note that $H^t$ is also a Hadamard matrix.
By taking scalar products between columns of $H$ and its first column, which is all-one, we see that $ \sum_{x \in X} x_i=0$ for $i>1.$
That $ \sum_{x \in X} x_1=0,$ follows by the well-known property of the roots of unity that $\sum_{k=0}^{q-1}\zeta_q^k=0.$ Hence the said spherical code is a $1$-design.

The result follows by Table 1 of \cite{FL}.
\end{proof}

{\de A  spherical code $X$ is a {\bf 2-design} if it is  a  {\bf 1-design}, and if, furthermore, for all pairs $i \neq j$ of coordinate indices the following two relations hold.
$$\sum_{x \in X} (x_i^2-x_j^2)=0, \;\sum_{x \in X} x_ix_j=0.$$
}

{\thm If $H\in BH(n,q)$ is dephased, then $\phi(\zeta_q^{C_H})$ is a { 2-design} and its covering radius is at most $\sqrt{2(1-\frac{1}{2n})}.$}

\begin{proof}
We have seen in  Proposition  8.3 that $\phi(\zeta_q^{C_H})$ is a {1-design}.
Denote by $h_j$ the column number $j$ of $H.$ Write $w=\zeta_q.$
That
$$\sum_{x \in X} x_i^2=\sum_{x \in X} x_j^2=q$$ is immediate by the fact that the columns of $H$ have the same norm, by the Hadamard property of $H^t.$ Further, still using that property we get
$$ \sum_{x \in X} x_ix_j=\sum_{t=0}^{q-1}(\psi (w^t h_i),\psi(w^t h_j))=\frac{1}{{n}} \sum_{t=0}^{q-1} \Re(\langle h_i,h_j \rangle)=0.$$
Thus $\phi(\zeta_q^{C_H})$ is a { 2-design}. The upper bound follows then by \cite[Table 1]
{FL}.
\end{proof}
This bound can be sharpened under some extra conditions.

{\thm If $H\in BH(n,q),$ with $q$ even, is dephased, then $\phi(\zeta_q^{C_H})$ is an antipodal { 2-design} and its covering radius is at most $\sqrt{2(1-\frac{1}{\sqrt{2n}})}.$}

\begin{proof}
We see by Theorem 8.1 that $\phi(\zeta_q^{C_H})$ is a {2-design}. It is antipodal by the evenness of $q.$ The result follows then by Theorem 3 of \cite{S}.
\end{proof}

To compare with the results on the previous subsection we can give the following corollary.

{\cor \label{up} If $H\in BH(n,q),$ with $q$ even, is dephased, then the covering radius of $C_H$ for the Chinese Euclidean distance is bounded above as follows.
$$r_{CE}(C_H)\le 2n -\sqrt{2n} $$}

\begin{proof} Immediate upon considering the normalization of the spherical code. \end{proof}
\section{Conclusion and open problems}
In this paper we have explored a new definition of self-dual bent sequences in relation with Butson Hadamard matrices. The strong group of automorphisms which preserves the set of bent sequences for a given multiplier has been introduced, and given a generation algorithm. We have given computational methods to construct, for a given matrix, such a sequence. We have also given infinite families of Butson Hadamard matrices with attending self-dual bent sequences. We have also given an arithmetic criterion that allows to rule out all the matrices of $BH(n,q)$ from admitting eigenvalues allowing to define such sequences. In a second part, we have explored the code attached to the Butson Hadamard matrix over a finite ring and its concomitant spherical code from the viewpoint of the Chinese Euclidean distance, and of the standard Euclidean distance. The existence of bent sequences implies then a lower bound on the covering radius of these codes.
Upper bounds on the covering radius of the spherical code can be provided by showing its strength as a spherical design. The main open problem left is then to know for $n \to \infty$
which bound is closer to the true value of the covering radius of $C_H,$ the lower bound of Corollary \ref{low} or the upper bound of Corollary \ref{up}? A similar question was solved
for the Sylvester matrix and the generalized Sylvester matrix in \cite{Kai,Kaiq}, using difficult techniques from number theory, probability, and  discrepancy theory. In our case, Tables 4, 5, 6 seem to indicate that the true value equals the lower bound.\\ \ \

{\bf Acknowledgement:} The authors are indebted to Denis Krotov for helpful discussions.

			\end{document}